\documentclass[12pt]{article}
\usepackage{amssymb}
\usepackage{amsmath}
\usepackage{amsthm}
\voffset = - 1.2cm
\begin{document}
\font\frak=eufm10 scaled\magstep1
\font\fak=eufm10 scaled\magstep2
\font\fk=eufm10 scaled\magstep3
\font\black=msbm10 scaled\magstep1
\font\bigblack=msbm10 scaled\magstep 2
\font\bbigblack=msbm10 scaled\magstep3
\font\scriptfrak=eufm10
\font\tenfrak=eufm10
\font\tenblack=msbm10

%A continuacion definimos los comandos para utilizar los
%fuentes en modo matematico

%Operadores especiales, abrev. matematicasxs
\def\biggoth #1{\hbox{{\fak #1}}}
\def\bbiggoth #1{\hbox{{\fk #1}}}\def\sp #1{{{\cal #1}}}
\def\goth #1{\hbox{{\frak #1}}}
\def\scriptgoth #1{\hbox{{\scriptfrak #1}}}
\def\smallgoth #1{\hbox{{\tenfrak #1}}}
\def\smallfield #1{\hbox{{\tenblack #1}}}
\def\field #1{\hbox{{\black #1}}}
\def\bigfield #1{\hbox{{\bigblack #1}}}
\def\bbigfield #1{\hbox{{\bbigblack #1}}}
\def\Bbb #1{\hbox{{\black #1}}}
\def\v #1{\vert #1\vert}             %Para denotar elgrado de #1
\def\ord#1{\vert #1\vert}
\def\m #1 #2{(-1)^{{\v #1} {\v #2}}} %Para denotar el signo (-1)^...
\def\lie #1{{\sp L_{\!#1}}}               %%Lie derivative
\def\pd#1#2{\frac{\partial#1}{\partial#2}}
\def\pois#1#2{\{#1,#2\}}
 %  un parntesis de Poisson {f,g}
\def\set#1{\{\,#1\,\}}             %  notacin para conjuntos
\def\<#1>{\langle#1\rangle}        %  una forma bilineal <x,a>
\def\>#1{{\bf #1}}                %  notacin para vectores
\def\f(#1,#2){\frac{#1}{#2}}
\def\cociente #1#2{\frac{#1}{#2}}
\def\braket#1#2{\langle#1\mathbin\vert#2\rangle} %% <w|z>
\def\brakt#1#2{\langle#1\mathbin,#2\rangle}           %% <w,z>
\def\dd#1{\frac{\partial}{\partial#1}}
\def\bra #1{{\langle #1 |}}
\def\ket #1{{| #1 \rangle }}
\def\ddt#1{\frac{d #1}{dt}}
\def\dt2#1{\frac{d^2 #1}{dt^2}}
\def\matriz#1#2{\left( \begin{array}{#1} #2 \end{array}\right) }
\def\Eq#1{{\begin{equation*} #1 \end{equation*}}}

%%Abreviaturas de simbolos
\def\bw{{\bigwedge}}      %%from Marmo diff geom
\def\hut{{\scriptstyle \land}}            %%from Marmo diff geom
\def\dg{{\goth g^*}}                                                                                                            %%dual of the Lie algebra
\def\Cdg{{C^\infty (\goth g^*)}}
\def\poi{\{\:,\}}                           %parntesis de Poisson {,}
\def\qw{\hat\omega}                %  omega con sombrero
\def\FL{{\sp F}L}                 %  abreviatura para Transf. Legendre
\def\hFL{\widehat{{\sp F}L}}      %  abreviat. para Transf. Legendreext.
\def\XHMw{\goth X_H(M,\omega)}
\def\XLHMw{\goth X_{LH}(M,\omega)}
\def\ea{\varepsilon_a}
\def\ep{\varepsilon}
\def\mitad{\frac{1}{2}}
\def\x{\times}
\def\cinf{C^\infty}
\def\forms{\bigwedge}                 %  formas
\def\onda{\tilde}
\def\orb{{\sp O}}

%%letras griegas
\def\a{\alpha}
\def\d{\delta}
\def\g{{\gamma }}                  %  gama
\def\G{{\Gamma}}	
\def\La{\Lambda}                   %  lambda
\def\la{\lambda}                   %  Lambda
\def\w{\omega}                     %  una forma simplectica
\def\W{{\Omega}}                   %  Omega
\def\ltimes{\bowtie}

%%letras cal, Bbb, goticas, etc.
\def\roc{{\tilde{\cal R}}}                       %%from Marmo diff geom
\def\cl{{\cal L}}                               %%from Marmo diff geom
\def\V{{\sp V}}                                 %espacio de velocidades
\def\F{{\sp F}}
\def\cv{{{\goth X}}}                    %  un campo vectorial
\def\LG{\goth g}
\def\LH{\goth h}
\def\X{{{\goth X}}}                     %  un campovectorial
\def\R{{\hbox{{\field R}}}}             %%real numbers (Pepin)
\def\big R{{\hbox{{\bigfield R}}}}
\def\bbig R{{\hbox{{\bbigfield R}}}}
\def\C{{\hbox{{\field C}}}}         %%complex numbers (Pepin)
\def\Z{{\hbox{{\field Z}}}}             %%real numbers (Pepin)
\def\N{{\hbox{{\field N}}}}         %%complex numbers (Pepin)
%\def\small C{{\hbox{{\smallfield C}}}}         %%smallcomplex numbers (Pepin)

%%notaciones rm en modo math
\def\ima{\hbox{{\rm Im}}}                               %%Image of a map
\def\dim{\hbox{{\rm dim}}}        %%several definitions
\def\End{\hbox{{\rm End}}}
\def\Tr{\hbox{{\rm Tr}}}
\def\tr{{\hbox{\rm\small{Tr}}}}                %%Trace
\def\lin{{\hbox{Lin}}}
\def\vol{{\hbox{vol}}}
\def\Hom{{\hbox{Hom}}}
\def\rank{{\hbox{rank}}}
\def\Ad{{\hbox{Ad}}}
\def\ad{{\hbox{ad}}}
\def\CoAd{{\hbox{CoAd}}}
\def\coad{{\hbox{coad}}}
\def\Rea{\hbox{Re}}                     %  parte real
\def\id{{\hbox{id}}}                    %  la identidad
\def\Id{{\hbox{Id}}}
\def\Int{{\hbox{Int}}}
\def\Ext{{\hbox{Ext}}}
\def\Aut{{\hbox{Aut}}}
\def\Card{{\hbox{Card}}}
\def\SODE{{\small{SODE }}}
\newcommand{\bea}{\begin{eqnarray}}
\newcommand{\eea}{\end{eqnarray}}

\def\R{\mathbb{R}}
\def\ba{\begin{eqnarray}}
\def\ea{\end{eqnarray}}
\def\be{\begin{equation*}}
\def\ee{\end{equation*}}
\def\Eq#1{{\begin{equation*} #1 \end{equation*}}}
\def\R{\Bbb R}
\def\C{\Bbb C}
\def\Z{\Bbb Z}
\def\a{\alpha}                  % alpha
\def\b{\beta}                   % beta
\def\g{\gamma}                  % gamma
\def\d{\delta}                  % delta
\def\bra#1{\langle#1|}
\def\ket#1{|#1\rangle}
\def\goth #1{\hbox{{\frak #1}}}
\def\<#1>{\langle#1\rangle}
\def\cotg{\mathop{\rm cotg}\nolimits}
\def\Map{\mathop{\rm Map}\nolimits}
\def\wt{\widetilde}
\def\const{\hbox{const}}
\def\grad{\mathop{\rm grad}\nolimits}
\def\Div{\mathop{\rm div}\nolimits}
\def\braket#1#2{\langle#1|#2\rangle}
\def\Erf{\mathop{\rm Erf}\nolimits}
\def\matriz#1#2{\left( \begin{array}{#1} #2 \end{array}\right) }
\def\Eq#1{{\begin{equation*} #1 \end{equation*}}}
\def\deter#1#2{\left| \begin{array}{#1} #2 \end{array}\right| }
\def\pd#1#2{\frac{\partial#1}{\partial#2}}
\def\til{\tilde}

\def\la#1{\lambda_{#1}}
\def\teet#1#2{\theta [\eta _{#1}] (#2)}
\def\tede#1{\theta [\delta](#1)}
\def\N{{\frak N}}
\def\GR{{\cal G}}
\def\Wei{\wp}

\def\frac#1#2{{#1\over#2}} \def\pd#1#2{\frac{\partial#1}{\partial#2}}
                                                %  una derivada parcial
\def\matrdos#1#2#3#4{\left(\begin{matrix}#1 & #2 \cr          %para matrices 2X2
                                 #3 & #4 \cr\end{matrix}\right)}
%Si no se puede utilizar el fichero mssymb, los fuentes AmS TeX se
%pueden cargar a mano (por ejemplo) con las lineas siguientes
%\newfont{\got}{eufm10 scaled\magstep1}
%\newfont{\field}{msym10 scaled\magstep1}

\newtheorem{teor}{Teorema}[section]
\newtheorem{cor}{Corolario}[section]
\newtheorem{prop}{Proposici\'on}[section]
\newtheorem{note}[prop]{Note}
\newtheorem{definicion}{Definici\'on}[section]
\newtheorem{lema}{Lema}[section]
%\newexample{ejem}{Ejemplo}[section]
\theoremstyle{plain}
\newtheorem{theorem}{Theorem}
\newtheorem{corollary}{Corollary}
\newtheorem{proposition}{Proposition}
\newtheorem{definition}{Definition}
\newtheorem{lemma}{Lemma}

\def\Eq#1{{\begin{equation*} #1 \end{equation*}}}
\def\R{\Bbb R}
\def\C{\Bbb C}
\def\Z{\Bbb Z}
\def\mp#1{\marginpar{#1}}

\def\la#1{\lambda_{#1}}
\def\teet#1#2{\theta [\eta _{#1}] (#2)}
\def\tede#1{\theta [\delta](#1)}
\def\N{{\frak N}}
\def\Wei{\wp}
\def\Hil{{\cal H}}

\font\frak=eufm10 scaled\magstep1

\def\bra#1{\langle#1|}
\def\ket#1{|#1\rangle}
\def\goth #1{\hbox{{\frak #1}}}
\def\<#1>{\langle#1\rangle}
\def\cotg{\mathop{\rm cotg}\nolimits}
\def\cotanh{\mathop{\rm cotanh}\nolimits}
\def\arctanh{\mathop{\rm arctanh}\nolimits}
\def\wt{\widetilde}
\def\const{\hbox{const}}
\def\grad{\mathop{\rm grad}\nolimits}
\def\Div{\mathop{\rm div}\nolimits}
\def\braket#1#2{\langle#1|#2\rangle}
\def\Erf{\mathop{\rm Erf}\nolimits}

\centerline{\Large \bf Integrability of Lie systems}
\vskip 0.75cm

\centerline{\Large \bf through Riccati equations}
\vskip 0.75cm

\centerline{ Jos\'e F. Cari\~nena$^\dagger$ and Javier de Lucas$^{\dagger\ddagger}$}
\vskip 0.5cm

\centerline{$^{\dagger}$Departamento de  F\'{\i}sica Te\'orica, Universidad de Zaragoza,}
\medskip
\centerline{50009 Zaragoza, Spain.}
\medskip
\centerline{$^{\ddagger}$Institute of Mathematics of the Polish Academy of Science,}
\medskip
\centerline{P.O. Box 00-956, Warszawa, Poland.}
\medskip

\begin{abstract}
Integrability conditions for Lie systems are related to reduction or transformation processes.
We here analyse a geometric method to construct
integrability conditions for Riccati equations following these approaches. Our procedure provides us with
 a unified geometrical viewpoint that allows us to analyse some previous
works on the topic and explain new properties. Moreover, this new approach can be straightforwardly
 generalised to describe integrability conditions for any Lie
 system. Finally
 we show the usefulness of our treatment in order to study the problem of the
 linearisability of
 Riccati equations.
\end{abstract}
\vskip 2.0cm
{{\bf Keywords: }Lie-Scheffers systems; Riccati equations; integrability conditions.}\\
\noindent
{{\bf 2000 Mathematics Subject Classification:} 17B80, 34A05, 34A26, 34A34}

\section{Introduction}
\indent

The Riccati equation
\begin{equation}\label{ricceq}
\frac{dy}{dt}=b_0(t)+b_1(t)y+b_2(t)y^2,
\end{equation}
is the simplest non-linear differential equation \cite{CRL07,CarRamGra} and it appears
 in many different fields of Mathematics and Physics \cite{CL09,CLR10,CMN,Ch09,MH05,RM08,SHC07,DS08,PW}. It is essentially the only first-order ordinary differential
equation on the real line admitting a
non-linear superposition principle \cite{LS,PW} and in spite of its apparent simplicity,
its general solution cannot be described by means
of quadratures except in some very particular cases \cite{AS64,CarRam,Ib08,Ib09II,Kamke,Ko06,Zh98,Zh99,Mu60,Na99,Na00,Pr80,Ra61,Ra62,RU68,RDM05,Stre}.

In this paper we review the geometric approach to Riccati equations
 according to the results of the works \cite{CRL07,CarRam} with the aim of proving that
integrability conditions of Riccati equations can be understood in a very
general way from the point of view of the theory of Lie systems
\cite{CR02,Gu93,Ib96,Ib99,Ib89,LS, Ve93,PW}. Furthermore, we
recover various known results as particular cases of our approach and the method here derived can be applied to any other Lie system, e.g. \cite{CL09,CLR10}.

Each Lie system is associated with a Lie algebra of vector fields, the so-called Vessiot--Guldberg Lie algebra \cite{CLR10,GGL08,Gu93,Ib96,Ve93}. This Lie algebra can be used to classify those Lie systems that can be integrated by quadratures \cite{CarRamGra}. For instance, it is a known fact that Lie systems related to solvables Vessiot--Guldberg Lie algebras, e.g. affine homogeneous systems or linear homogeneous systems, can be integrated by quadratures \cite{CarRamGra,Ib09II,Ib09}. Nevertheless, the general solution of Lie systems related to non-solvables Lie algebras, e.g. Riccati equations, cannot be completely determined and it frequently relies on the knowledge of certain special functions \cite{Ib09II}, the solution of other equations \cite{Na99,Na00,DS08}, etc. 

The method developed here allows us to determine integrable cases of Lie systems related to non-solvables Vessiot--Gulberg Lie algebras. Such a procedure is detailed for Riccati equations, which are associated with a non-solvable Vessiot--Guldberg Lie algebra isomorphic to $\mathfrak{sl}(2,\mathbb{R})$, but it can also be applied to other Lie systems related to Vessiot--Guldberg Lie algebras isomorphic to this one, for example, the Lie systems connected to Milne--Pinney equations \cite{CL09Milne}, Ermakov systems \cite{CLR08e}, harmonic oscillators \cite{CLR10}, etc. 

We finally analyse the linearisation of Riccati equation by means of our new approach and recover a characterisation previously proved by Ibragimov \cite{Ib08}. Furthermore, we detail a, as far as we know, new result about the properties of linearisation of Riccati equations. 

The paper is organised as follows. For the sake of completeness,
we report some known facts on the integrability of Riccati equations in Sec. 2 and we review the geometric interpretation of
the general Riccati equation as a $t$-dependent vector field on
the one-point compactification of the real line in Sec. 3. As a consequence of the latter, Riccati equations can be studied through equations on $SL(2,\mathbb{R})$.
Sec. 4 is devoted to reporting some known results on the action of the group of curves in $SL(2,\mathbb{R})$
on the set of Riccati equations and how this action can be seen in terms of
transformations
of the corresponding equations on $SL(2,\mathbb{R})$, see \cite{CarRam05b}. In Sec. 5 we build up a Lie system describing the transformation process of Riccati
equations through the action of curves of $SL(2,\mathbb{R})$ and then, in Sec. 6, we analyse
the general characteristics of our approach into integrability conditions and
how the transformation processes described by the previous Lie system can be used to give a unified approach to the results of \cite{CRL07,CarRamGra}. In Sec. 7 we develop
a particular case of the procedures of Sec. 6 in order to recover some results found in the literature \cite{AS64,Ko06,Ra61,Ra62,RU68,RDM05}. Sec. 8 is devoted to analyzing the theory of integrability through reduction from our new viewpoint. Finally, in
Sec. 9 we describe how our Lie system for studying
integrability conditions enables us  to explain when certain linear fractional transformations allow us to linearise Riccati equations. As a particular instance we obtain a result given in \cite{Ib08,RDM05}.

\section{Integrability of Riccati equations}\label{IntRicEqu}
\indent

In order to provide a first insight into the study of integrability conditions for Riccati equations, and for any Lie system in general, we review in this Section some known results about the integrability of Riccati equations.

As a first particular example, Riccati equations (\ref{ricceq}) are integrable by quadratures when
$b_2=0$. Indeed, in such a case these equations reduce to an inhomogeneous linear equation
and two quadratures allow us to find the
general solution.

Additionally, under the change of variable
$w=-1/y$  the  Riccati  equation
(\ref{ricceq}) reads
$$
\frac{dw}{dt}=b_0(t)\, w^2-b_1(t)\,w+b_2(t)
$$
and if we suppose  $b_0=0$ in Eq. (\ref{ricceq}),
 then the mentioned change of variable transforms the given equation into an integrable
linear one.

Another very well-known property on integrability of Riccati equations is
that given a particular solution $y_1(t)$ of Eq. (\ref{ricceq}),
then the change of variable $y=y_1(t)+z$ leads to a new Riccati  equation for which the coefficient of the
 term independent of $z$
 is zero, i.e.
 \begin{equation}
\frac {dz}{dt}=[2\, b_2(t)\, y_1(t)+ b_1(t)] z+ b_2(t)\,z^2, \label{Bereq}
\end{equation}
and, as we pointed out before, it can be reduced to an inhomogeneous  linear equation with the change
$z=-1/u$. Therefore, given one particular solution, the general
solution can be found by means of two quadratures.

 If not only one but two particular solutions,
$y_1(t)$ and $y_2(t)$, of Eq. (\ref{ricceq}) are known, the general solution can be found by means of only one quadrature. In fact, the change of variable
$z=(y-y_1(t))/(y-y_2(t))$ transforms the original equation into a homogeneous first-order linear differential equation in the new variable  $z$
and therefore the general solution can immediately be found.

 Finally, giving three particular solutions, $y_1(1),y_2(t),y_3(t)$, the general solution can be
 written, without making use of any quadrature, in the following way
$$
y(t)=\frac{y_1(t)(y_3(t)-y_2(t))-ky_2(t)(y_1(t)-y_3(t))}{(y_3(t)-y_2(t))-k(y_1(t)-y_3(t))}.
$$
This is a non-linear superposition rule studied  in
\cite{CMN} from a group theoretical perspective.

The simplest case of Eq. (\ref{ricceq}), when it is an autonomous equation ($b_0$, $b_1$ and $b_2$ constants), has been fully studied (see
e.g. \cite{CarRamdos} and references therein) and it is integrable by
quadratures. This result can be considered as a consequence of the existence of a constant
(maybe complex) solution enabling us to reduce the Riccati equation into an
inhomogeneous linear one. Moreover, the separable Riccati equations of the form
\begin{equation*}
\frac{dy}{dt}=\varphi(t)(c_0+c_1\,y+c_2\,y^2),
\end{equation*}
with $\varphi(t)$ a non-vanishing function on a certain open interval $I\subset \mathbb{R}$ and $c_0$, $c_1$, $c_2$ real numbers, are integrable because a new time function $\tau=\tau(t)$
such that $d\tau/dt=\varphi(t)$ reduces the above equation into an autonomous one. Furthermore, the above Riccati equations are also integrable as they accept, in similarity to the autonomous case, a constant (maybe complex) solution.

\section{Geometric approach to Riccati equations}
\noindent

Let us report in this Section some known results about the geometrical approach to the Riccati equation \cite{CRL07}. Such a point of view is used in next Sections to investigate integrability conditions for these equations and, in general, for any Lie system.

 From the geometric viewpoint, the Riccati equation
(\ref{ricceq})
 can be considered
as a differential equation determining the integral curves for  the
$t$-dependent  vector field \cite{Car96}
\begin{equation}
X(t,y)=\left[b_0(t)+b_1(t)y+b_2(t)y^2\right]\frac{\partial}{\partial y}\ .\label{vfRic}
\end{equation}

This $t$-dependent vector field is a linear combination with $t$-dependent coefficients $b_0(t)$, $b_1(t)$ and $b_2(t)$
of the three  vector fields
\begin{equation}
L_0 =\frac{\partial}{\partial y}\,,        \quad
L_1 =y\,\frac{\partial}{\partial y}\, ,    \quad
L_2 = y^2\,\frac{\partial}{\partial y}\,,  \label{sl2gen}
\end{equation}
with defining relations
\begin{equation}\label{conmutL}
[L_0,L_1] = L_0\,,      \quad
[L_0,L_2] = 2L_1\,,     \quad
[L_1,L_2] = L_2 \,,
\end{equation}
and therefore spanning a three-dimensional Lie algebra of vector fields $V$. Consequently, Riccati equations are Lie systems \cite{LS} and the Lie algebra $V$, the so-called Vessiot--Guldberg Lie algebra \cite{Gu93,Ve93}, is isomorphic to  $\mathfrak{sl}(2,\mathbb{R})$ being here considered as made up by traceless $2\times 2$ matrices. A particular basis for $\mathfrak{sl}(2,\mathbb{R})$ is given by
\begin{equation}
M_0=\left(\begin{matrix}
0&-1\\0&0     
\end{matrix}\right)\,,
M_1=\frac{1}{2}\left(\begin{matrix}
-1&0\\0&1               
\end{matrix}\right)\,,
M_2=\left(\begin{matrix}
 0&0
\\1&0    
    \end{matrix}\right)\ .
\label{base_matrices}
\end{equation}
Moreover, it can be checked that the linear map $\rho:\mathfrak{sl}(2,\mathbb{R})\rightarrow V$ obeying $\rho(M_j)=L_j$, with $j=0,1,2$, is
 a Lie algebra isomorphism.

Note that $L_2$ is not a complete vector field on $\mathbb{R}$. However we can do the one-point compactification $\overline{\mathbb{R}}=\mathbb{R}\cup \{\infty\}$ of $\mathbb{R}$ and then $L_0$, $L_1$ and  $L_2$ are complete vector fields on $\overline{\mathbb{R}}$. Consequently, these vector fields are fundamental vector fields
corresponding to the action  $\Phi:(A,y)\in SL(2,{\mathbb{R}})\times \overline{\mathbb{R}}\mapsto \Phi(A,y)\in\overline{\mathbb{R}}$ given by
\begin{equation}\label{Action}
\Phi(A,y)=\left\{
\begin{aligned}
&{\frac{\alpha\, y+\beta}{\gamma\, y+\delta}}\quad &y&\neq-{\frac{\delta}{\gamma}},\,\,\,y\neq\infty,\\
&\frac{\alpha}{\gamma} \quad  &y&=\infty,\\
&\infty \quad&y&=-\frac{\delta}{\gamma}, \\
\end{aligned}
\right.\quad {\rm with}\quad A=\left(\begin{array}{cc}\alpha& \beta\\\gamma&\delta\end{array}\right)\in SL(2,\mathbb{R}).
\end{equation}

Denote by $X^{\tt R}_j$ and  $ X^{\tt L}_j$, $j=1,2,3$, the right- and left-invariant vector fields on $SL(2,{\mathbb{R}})$ such that $X^{\tt
  R}_j(I)=X^{\tt
  L}_j(I)=M_j$. Moreover, these vector fields satisfy that $X^{\tt R}_j(A)=M_j\cdot A$ and $X^{\tt L}_j(A)=A\cdot M_j$, with  ``$\cdot$'' the usual matrix multiplication.

A remarkable property  is that if $A(t)$ is
 the integral curve for
the $t$-dependent
 vector field
$$X(t)=-\sum_{j=0}^2 b_j(t)\, X^{\tt R}_j\,,
$$
starting from   the
 neutral element  in $SL(2,\mathbb{R})$, i.e.  $A(0)=I$,  then $A(t)$ satisfies the equation
\begin{equation}\label{eLA}
\dot{A}(t)A^{-1}(t)=-\sum_{j=0}^2b_j(t)M_j\equiv {\rm a}(t),
\end{equation}
and the  solution of Riccati equation (\ref{ricceq}) with initial
condition $y(0)=y_0$ is given by $y(t)=\Phi(A(t),y_0)$ \cite{CR02}.

Note that the r.h.s. in Eq. (\ref{eLA}) is a curve in
$T_ISL(2,\mathbb{R})$ that can be identified to a curve in the Lie algebra
$\mathfrak{sl}(2,\mathbb{R})$ of left-invariant vector fields on $SL(2,\mathbb{R})$
through the usual isomorphism: we relate each left-invariant vector field $X^{\tt L}$ to
the element $X^{\tt L}(I)\in T_ISL(2,\mathbb{R})$. From now on, we do not distinguish explicitly
 elements in $T_ISL(2,\mathbb{R})$ and its corresponding ones in $\mathfrak{sl}(2,\mathbb{R})$.

In summary, the general solution of Riccati equations (\ref{ricceq}) can be
obtained  through solutions of an equation like (\ref{eLA}) starting from $I$. Consequently, we have reduced the problem of finding the general solution of Riccati equations to determining the solution of Eq.
(\ref{eLA}) beginning at the neutral element of $SL(2,\mathbb{R})$. Note that, in a similar way, this procedure can be applied to any Lie system \cite{CL09}.

\section{Transformation laws of Riccati equations}\label{TL}
\noindent

In this Section we briefly describe an important property of Lie systems, in the
particular case of Riccati equations, which plays a very relevant r\^ole for
establishing, as indicated in \cite{CarRam},  integrability
criteria: {\it The group $\mathcal{G}$ of curves in a Lie group $G$ associated with
a Lie system, here  $SL(2, {\mathbb{R}})$, acts on the set  of these  Lie systems, here
Riccati equations}.

More explicitly, fixed a basis of  vector fields on $\overline{\mathbb{R}}$, for instance  $\{L_j\,|\,j=0,1,2\}$, which spans a Vessiot--Guldberg Lie algebra of vector fields isomorphic to $\mathfrak{sl}(2,\mathbb{R})$,  each Riccati equation  (\ref{ricceq}) can be considered as a
curve $(b_0(t),b_1(t),b_2(t))$ in $\mathbb{R}^3$. The point now is that
 each element  of the group of smooth curves in $SL(2, \mathbb{R})$, i.e. $\bar A\in \mathcal{G}\equiv{\rm Map}(\mathbb{R},\,SL(2,\mathbb{R}))$,
transforms every curve $y(t)$ in $\overline{\mathbb{R}}$
 into a new curve $y'(t)$ in $\overline{\mathbb{R}}$ given by $y'(t)=\Phi(\bar A(t),y(t))$. Moreover, the $t$-dependent change of variables $y'(t)=\Phi(\bar A(t),y(t))$ transforms the Riccati equation (\ref{ricceq})
 into a new
 Riccati equation with new $t$-dependent coefficients, $b'_0,b'_1, b'_2$ given by

  \begin{equation}\label{trans}
\left\{\begin{aligned}
b'_2&={\delta}^2\,b_2-\delta\gamma\,b_1+{\gamma}^2\,b_0+\gamma {\dot{\delta}}-\delta \dot{\gamma}\ ,\\
b'_1&=-2\,\beta\delta\,b_2+(\alpha\delta+\beta\gamma)\,b_1-2\,\alpha\gamma\,b_0
       +\delta \dot{\alpha}-\alpha \dot{\delta}+\beta \dot{\gamma}-\gamma \dot{\beta}\ ,   \\
b'_0&={\beta}^2\,b_2-\alpha\beta\,b_1+{\alpha}^2\,b_0+\alpha\dot{\beta}-\beta\dot{\alpha},
\end{aligned}\right.
\end{equation}
with
$$
\bar{A}(t)=\left(
\begin{matrix}
\alpha(t)&\beta(t)\\
\gamma(t)&\delta(t)
\end{matrix}\right).
$$
The above transformation defines an affine action (see e.g.  \cite{LM87} for the general definition of
this concept) of the group
 $\mathcal{G}$ on the set of
Riccati equations, see \cite{CarRam}.

The group $\mathcal{G}$ also acts on the set of equations of the form (\ref{eLA}) on $SL(2,\mathbb{R})$. In order to show this, note first that $\mathcal{G}$  acts on the left on the
set of curves in $SL(2,\mathbb{R})$ by left translations, i.e. given two curves $A(t)$ and $\bar A(t)$ in $SL(2,\mathbb{R})$, the curve $\bar A(t)$ transforms the curve $A(t)$  into a new one  $A'(t)=\bar A(t) A(t)$. Moreover, if
 $A(t)$ is a solution of Eq. (\ref{eLA}), then the new curve $A'(t)$ satisfies a new equation like (\ref{eLA}) but  with a different right
hand side ${\rm a}'(t)$. Differentiating the relation $A'(t)=\bar A(t) A(t)$ in terms of time and taking into account the form of (\ref{eLA}), we get that the relation between the curves ${\rm a}(t)$ and ${\rm a}'(t)$ in $\mathfrak{sl}(2,\mathbb{R})$ is
\begin{equation}
{\rm a}'(t)=\bar A(t){\rm a}(t)\bar A^{-1}(t)+\dot{\bar{A}}(t)\bar A^{-1}(t)
=-\sum_{j=0}^2b'_j(t)M_j\, \label{newricc}
\end{equation}
and such a relation implies the expressions (\ref{trans}). Conversely,  if  $A'(t)=\bar A(t) A(t)$ is  the solution for the  equation corresponding to the
curve ${\rm a}'(t)$ given by the transformation rule (\ref{newricc}), then $A(t)$ is the solution of Eq. (\ref{eLA}).

To sum up, we have shown that it is possible to associate each Riccati equation with an equation on
the Lie group $SL(2,\mathbb{R})$ and to define an infinite-dimensional group of 
transformations acting on the set of Riccati equations. Additionally, this
process 
can be easily derived in a similar way for any Lie system. In such a case, we
must consider an equation on a Lie group $G$ associated with the corresponding Lie
 system 
and the group $\mathcal{G}$ of curves in $G$ acting on the set of curves in $G$
in the
 form $A'(t)=L_{\bar A(t)}A(t)$ instead of $A'(t)=\bar A(t)A(t)$. This action
 induces
 other action of $\mathcal{G}$ on the set of equations of the form (\ref{eLA})
 but on the Lie group $G$. 
More explicitly, a curve $\bar A(t)\in \mathcal{G}$ transforms an equation on
$G$ of the form (\ref{eLA}) 
determined by a curve ${\rm a}(t)\subset T_IG$ into a new one determined by the
new curve 
${\rm a}'(t)\subset T_IG$ given by
\begin{equation}\label{TransConecc}
{\rm a}'(t)={\rm Ad}_{\bar A(t)}{\rm a}(t)+R_{\bar A^{-1}(t)*\bar A(t)}\dot{\bar{A}}(t).
\end{equation}

\section{Lie structure of an equation of transformation of Lie systems}
\indent

Our aim in this Section is to  construct a Lie system describing the curves in $SL(2,\mathbb{R})$
relating two Riccati equations associated with a pair of equations in
$SL(2,\mathbb{R})$ characterised by two 
curves ${\rm a}(t), {\rm
  a}'(t)\subset \mathfrak{sl}(2,{\mathbb{R}})$. By means of this Lie system we are
going to  explain  in next Sections the developments of \cite{CRL07, CarRamGra} and other works from a unified viewpoint.

Let us multiply Eq. (\ref{newricc}) on the right by $\bar A(t)$ to get
\begin{equation}\label{MatrixRicc}
\dot{\bar{A}}(t)={\rm a}'(t)\bar A(t)-\bar A(t){\rm a}(t)\,.
\end{equation}
If we consider Eq. (\ref{MatrixRicc}) as a first-order differential equation in the
coefficients of the
 curve $\bar A(t)$ in $SL(2,\mathbb{R})$, with
$$
\bar A(t)=\left(
\begin{matrix}
\alpha(t) &\beta(t)\\
\gamma(t)& \delta(t)
\end{matrix}\right)\,,\quad \alpha(t)\delta(t)-\beta(t)\gamma(t)=1,
$$
then system (\ref{MatrixRicc}) reads
\begin{equation}\label{FS}
\left(\begin{matrix}
\dot\alpha\\
\dot\beta\\
\dot\gamma\\
\dot\delta
\end{matrix}\right)
=
\left(\begin{matrix}
\frac{b'_1-b_1}{2}&b_2 &b'_0&0\\
-b_0& \frac{b'_1+b_1}{2}&0 &b'_0\\
-b'_2&0 &-\frac{b'_1+b_1}{2}& b_2\\
0&-b_2' &-b_0& -\frac{b'_1-b_1}{2}
\end{matrix}\right)
\left(\begin{matrix}
\alpha\\
\beta\\
\gamma\\
\delta
\end{matrix}\right).
\end{equation}

In order to determine the solutions $x(t)=(\alpha(t),\beta(t),
\gamma(t),\delta(t))$ of the above system relating two
different Riccati equations,  we should check that actually  the matrices $\bar A(t)$,
whose elements are the corresponding components of $x(t)$, are related to
matrices in $SL(2,\mathbb{R})$, i.e.
we have to verify that at any time
$\alpha\delta-\beta\gamma=1$. Nevertheless, we can drop such a restriction because it can be automatically implemented by a restraint on the initial conditions for
 the solutions and hence we can deal with the variables $\alpha,\beta,\gamma, \delta$ in the
 system  (\ref{FS}) as being independent. Consider now the vector fields
{\small
\begin{equation*}
\begin{array}{ll}
N_0=-\alpha\dfrac{\partial}{\partial\beta}-\gamma\dfrac{\partial}{\partial\delta}, &N'_0=\gamma\dfrac{\partial}{\partial\alpha}+\delta\dfrac{\partial}{\partial\beta},\cr
N_1=\frac
12\left(\beta\dfrac{\partial}{\partial\beta}+\delta\dfrac{\partial}{\partial\delta}-\alpha\dfrac{\partial}{\partial\alpha}-\gamma\dfrac{\partial}{\partial\gamma}\right),
&N'_1=\frac 12\left(\alpha\dfrac{\partial}{\partial\alpha}+\beta\dfrac{\partial}{\partial\beta}-\gamma\dfrac{\partial}{\partial\gamma}-\delta\dfrac{\partial}{\partial\delta}\right),\cr
N_2=\beta\dfrac{\partial}{\partial\alpha}+\delta\dfrac{\partial}{\partial\gamma},& N'_2=-\alpha\dfrac{\partial}{\partial\gamma}-\beta\dfrac{\partial}{\partial\delta},\nonumber
\end{array}
\end{equation*}}
satisfying the non-null commutation relations
\begin{eqnarray*}
&&\left[ N_0,N_1\right]=N_0, \qquad [N_0,N_2]=2  N_1, \qquad [N_1,N_2]=N_2,\cr
&&[N'_0,N'_1]=N'_0, \qquad [N'_0, N'_2]=2
  N'_1,\qquad [N'_1,  N'_2]=N'_2\,.\nonumber
\end{eqnarray*}
Note that as $[N_i,N'_j]=0$, for $i,j=0,1,2$, the linear system of differential equation (\ref{FS}) is a Lie system on $\mathbb{R}^4$  associated with a Lie algebra of vector fields isomorphic to $\mathfrak{g}\equiv\mathfrak{sl}(2,\mathbb{R})\oplus\mathfrak{sl}(2,\mathbb{R})$. This Lie
algebra decomposes into a direct sum of two Lie algebras of vector fields isomorphic to
$\mathfrak{sl}(2,\mathbb{R})$: the first one is spanned by $\{N_0,N_1,N_2\}$ and the second
one by $\{N'_0,N'_1,N'_2\}$.

If we  denote $x\equiv\left(\alpha,\beta,\gamma,\delta\right)\in \mathbb{R}^4$,  the
system (\ref{FS}) is a differential equations on $\mathbb{R}^4$
\begin{equation*}
\frac{dx}{dt}=N(t,x),
\end{equation*}
with $N$ being the $t$-dependent  vector field
\begin{equation*}
N(t,x)=\sum_{j=0}^2\left(b_\alpha(t)N_\alpha(x)+b'_\alpha(t)N'_\alpha(x)\right).
\end{equation*}

The vector fields $\{N_0,N_1,N_2,N'_0,N'_1,N'_2\}$ span a regular involutive distribution $\mathcal{D}$ with rank
three in almost any point of $\mathbb{R}^4$ and thus there exists, at least locally, a
first-integral. We can check that the function 
$$I:x\equiv (\alpha,\beta,\gamma,\delta)\in\mathbb{R}^4\longrightarrow I(x)\equiv\det x\equiv\alpha\delta-\beta\gamma \in \mathbb{R}$$
is a first-integral for the vector fields in the distribution $\mathcal{D}$. Moreover, such a first-integral is related to the
 determinant of the matrix $\bar A$ with coefficients given by the components
 of $x=(\alpha,\beta,\gamma,\delta)$. Therefore, if we have a solution of the system (\ref{FS})
with an initial condition $\det x(0)=\alpha(0)\delta(0)-\beta(0)\gamma(0)=1$,
then $ \det x(t)=1$  at any time $t$ and the solution can be understood
as a curve in $SL(2,\mathbb{R})$.

In summary, we have proved that:

\begin{theorem}\label{THLS} The curves in $SL(2,\mathbb{R})$ transforming equation (\ref{eLA}) into a new equation of the same form but characterised by a new curve ${\rm a}'(t)=-\sum_{j=0}^2b'_j(t)M_j\,$  are described through the solutions of the Lie system
\begin{equation}\label{Sys}
\frac{dx}{dt}=N(t,x)\equiv\sum_{j=0}^2\left(b_j(t)N_j(x)+b'_j(t)N'_j(x)\right)\,
\end{equation}
such that $\det x(0)=1$. Furthermore, the above Lie system is related to a non-solvable Vessiot--Guldberg Lie algebra isomorphic to $\mathfrak{sl}(2,\mathbb{R})\oplus\mathfrak{sl}(2,\mathbb{R})$.
\end{theorem}
\begin{corollary} \label{CorCur} Given two Riccati equations associated with
  curves ${\rm a}'(t)$ and ${\rm a}(t)$ in
$\mathfrak{sl}(2,\mathbb{R})$ there always exists a curve $\bar A(t)$ in $SL(2,\mathbb{R})$
transforming the Riccati equation related to ${\rm a}(t)$ into the one associated with ${\rm a}'(t)$. If
furthermore  $\bar A(0)=I$, this curve is uniquely defined.
\end{corollary}
\begin{proof}
Given a matrix $A(0)\in SL(2,\mathbb{R})$ and an element $x(0)$ related to it, according to the theorem of existence and uniqueness of solutions for differential equations, the system (\ref{Sys}), with the chosen ${\rm a}'(t)$ and ${\rm a}(t)$, admits a solution $x(t)$ with initial condition $x(0)$. As ${\rm det}(x(t))={\rm det}(x(0))=1$, such a solution, considered as a matrix $\bar
A(t)$, belongs to
 $SL(2,\mathbb{R})$ and therefore  there exists a solution $x(t)$ for the system (\ref{Sys}) with initial condition $x(0)$ related to $\bar A(0)$. This proves the first statement of our corollary.

If the curve  $\bar A(t)$ connecting two curves in $\mathfrak{sl}(2,\mathbb{R})$
satisfies  $\bar A(0)=I$, it is the curve $x(t)$ in $\mathbb{R}^4$
being the solution
of system (\ref{Sys}) with initial condition $x(0)=(1,0,0,1)$, which is
uniquely determined because of  the theorem of existence and uniqueness of
solutions of systems of first-order differential equations.
\end{proof}

Even if we know that given two equations on the Lie group $SL(2,\mathbb{R})$ there
always exists a transformation relating both, in order to obtain such a curve
 we need to solve the Lie system (\ref{Sys}). Unfortunately, such a Lie system is associated with a non-solvable Lie algebra and it is
not easy in general to find its solutions, i.e. it is not integrable by
quadratures and therefore such a  curve  cannot be easily found
in the general case.

Nevertheless, we will explain many known properties and obtain new integrability conditions for Riccati equations by means of Theorem  \ref{THLS}. Furthermore, the procedure to obtain the Lie system (\ref{Sys}) can be generalised to deal with any Lie system related to a Lie group $G$ with Lie algebra $\mathfrak{g}$. In this general case, relation (\ref{TransConecc}) implies that
\begin{equation*}
\dot{\bar{A}}(t)=R_{\bar{A}(t)*I}{\rm a}'(t)-L_{\bar{A}(t)*I}{\rm a}(t). 
\end{equation*}
As $X^{\tt R}(t,\bar A)=R_{\bar{A}*I}{\rm a}'(t)$ is a $t$-dependent
right-invariant vector field 
on $G$ and $X^{\tt L}(t,\bar A)=-L_{\bar{A}*I}{\rm a}(t)$ a left-invariant one,
the above system 
is the equation determining the integral curves of a time-dependent vector
field with values in 
the linear space spanned by right- and left- invariant vector fields on
$G$. Note that the family
 of left-invariant (right-invariant) vector fields on $G$ spans a Lie algebra
 isomorphic to 
$\mathfrak{g}$ and, as right- and left-invariant vector fields commute among
them, the set 
of vector fields spanned by both families is a Lie algebra of vector fields
isomorphic to 
$\mathfrak{g}\oplus\mathfrak{g}$. In this way, we get that the above system,
relating two Lie
 systems associated with curves ${\rm a}(t)$ and ${\rm a}'(t)$ in
 $\mathfrak{g}$, 
is a Lie system related to a Vessiot--Guldberg Lie algebra isomorphic to $\mathfrak{g}\oplus\mathfrak{g}$.

\section{Lie systems and integrability conditions}
\indent

In this section some integrability conditions are  analysed
from the perspective of the theory
of Lie systems  with $SL(2,\mathbb{R})$ as associated Lie group, with the aim of giving a unified approach to the reduction and
transformations procedures described in \cite{CRL07, CarRamGra}. More explicitly, these
methods are related to conditions for the existence of a curve in a previously
chosen family of curves in $SL(2,\mathbb{R})$ connecting a curve ${\rm a}(t)\subset\mathfrak{sl}(2,\mathbb{R})$ with  a curve ${\rm a}'(t)$
in a
 solvable Lie subalgebra of $\mathfrak{sl}(2,\mathbb{R})$. It is also shown that
 this viewpoint
enables us to explain many of the previous results scattered
 in the literature about this topic and to prove other new properties.

As it was shown in Sec. \ref{TL}, if the  curve $\bar A(t)\subset SL(2,\mathbb{R})$ transforms the equation on this Lie group defined by the  curve ${\rm a}(t)$ into
 another one characterised by ${\rm a}'(t)$ and $A'(t)$
is a solution for the equation similar to (\ref{eLA}) for the primed system, i.e. characterised by ${\rm a}'(t)$, then
$A(t)=\bar A^{-1}(t)A'(t)$ is a solution for the equation in $SL(2,\mathbb{R})$ characterised by ${\rm a}(t)$. Moreover, if ${\rm a}'(t)$ lies in a solvable Lie subalgebra of $\mathfrak{sl}(2,\mathbb{R})$,
we can obtain $A'(t)$ in many ways, e.g. by quadratures or by other methods
as those used in \cite{CarRamGra}. Then, once $A'(t)$ is obtained, the  knowledge  of the  curve $\bar A(t)$ transforming  the
 curve ${\rm a}(t)$ into ${\rm a}'(t)$ provides the curve $A(t)$.

Therefore if we begin with a curve
 ${\rm a}'(t)$ in a solvable Lie subalgebra of $\mathfrak{sl}(2,\mathbb{R})$ and consider the solutions for the system (\ref{Sys}) in a subset of $SL(2,\mathbb{R})$, we can
relate the curve ${\rm a}'(t)$, and therefore its Riccati equation, to other possible curves ${\rm a}(t)$, finding in this way
a family of Riccati equations that can be exactly solved. Note that, if we do not consider solutions of the system (\ref{Sys}) in a subset of $SL(2,\mathbb{R})$, it is generally difficult to check whether a particular Riccati equation belongs to the family of integrable Riccati equations so obtained. 

Suppose we impose some restrictions on the family of curves solutions of the system (\ref{Sys}), for instance  $\beta=\gamma=0$. Consequently, the system may not have solutions compatible with such restrictions, i.e. it may be impossible to connect the  curves ${\rm a}(t)$ and ${\rm a}'(t)$ by a curve in $SL(2,\mathbb{R})$ satisfying the assumed restrictions. This gives rise to some
compatibility conditions for the existence of these special solutions, some of them algebraic and
other differential
ones, between the $t$-dependent  coefficients of ${\rm a}'(t)$ and ${\rm
  a}(t)$. It will be shown later on that
such restrictions correspond to integrability conditions previously proposed  in the literature. 

Therefore, there are two ingredients to take into account:
\begin{enumerate}
 \item {\it The equations on the Lie group characterised by  curves ${\rm a}'(t)$ for which
we can obtain an explicit solution}.
We always suppose that ${\rm a}'(t)$ is related to a solvable Lie subalgebra of
$\mathfrak{sl}(2,\mathbb{R})$ and we leave open  other possible restrictions for further
study.
\item {\it The conditions imposed on the solutions of system} (\ref{FS}). We follow two principal approaches in next Sections
  where the
solutions of this system are related to curves in certain one-parameter or two-parameter subsets of $SL(2,\mathbb{R})$.
\end{enumerate}

Consider the next example of our theory: suppose we try to connect any ${\rm a}(t)$ with a final curve of the form ${\rm a}'(t)=-D(t)(c_0{\rm a}_0+c_1{\rm a}_1+c_2{\rm a}_2)$,
where $c_0,c_1$ and $c_2$ are real numbers. In this way, the
system (\ref{FS})  describing the curve $\bar A(t)\subset SL(2,\mathbb{R})$ connecting these curves is
\begin{equation}\label{Lie2}
\frac{dx}{dt}=\sum_{j=0}^2\left(b_j(t)N_j(x)+D(t)
c_j N'_j(x)\right)=N(t,x).
\end{equation}
Now, as the vector field
\begin{equation*}
N'=\sum_{j=0}^2c_j N'_j,
\end{equation*}
is such that
\begin{equation*}
\left[N_j,N'\right]=0,\quad\quad  j=0,1,2,
\end{equation*}
the Lie system (\ref{Lie2}) is related to a non-solvable Lie algebra of vector fields
 isomorphic to $\mathfrak{sl}(2,\mathbb{R})\oplus \mathbb{R}$. Hence, it is not integrable by
quadratures and the solution cannot be easily found in the general case. Nevertheless, note that system (\ref{Lie2}) always has a solution.

In this way, we can consider some particular cases of Lie system (\ref{Lie2}) for which the
resulting system of differential equations can be easily integrated. As a first
instance, take $x$ related to a one-parameter family of
elements of $SL(2,\mathbb{R})$. Such a restriction implies that system (\ref{Lie2}) has not
always a solution because sometimes it is not possible   to connect ${\rm
  a}(t)$ and ${\rm a}'(t)$ by means of the chosen
 family of curves. This fact  induces differential and/or algebraic restrictions on the
 initial $t$-dependent  functions $b_j$, with $j=0,1,2$, that
 describe some  known integrability conditions and may be some new ones developing the ideas of \cite{CRL07}. From this viewpoint we can  obtain new
 integrability conditions that  can be used, for instance, to obtain exact solutions.

Otherwise, if we choose a two-parameter set for the restriction, we  find in
some cases that we need a particular solution of the initial Riccati equation
to obtain the reduction of the given Riccati equation into an integrable one. This is the point of view shown in \cite{CarRamGra} where integrability conditions were related to reduction methods.

\section{Description of known integrability conditions}\label{DIC}
\indent

Let us first remark that Lie systems on $G$ of the form (\ref{eLA}) and determined by a constant curve,
${\rm a}=-\sum_{j=0}^2 c_j M_j$,
are integrable and consequently the same happens for
curves of the form ${\rm a}(t)=-D(t)\left(\sum_{j=0}^2 c_j M_j\right)$, where $D$ is any non-vanishing function, because a time-reparametrisation reduces the problem to
the previous one.

Our aim in this Section is to determine
 the curves $\bar A(t)$ in $SL(2,\mathbb{R})$ relating two
equations on $SL(2,\mathbb{R})$ characterised by the curves ${\rm a}(t)$ and ${\rm a}'(t)=-D(t)(c_0M_0+c_1M_1+c_2M_2)$ with $D(t)$ a non-vanishing function and $c_0$, $c_1$ and $c_2$ real constants such that $c_0c_2\neq 0$. As the final
equation is integrable, the transformation
establishing the relation to such a final integrable equation allows us to find
by quadratures  the solution of the initial equation and, therefore, the solution for its associated Riccati equation. In order to get such a
transformation, we look for curves $\bar A(t)$ in $SL(2,\mathbb{R})$ satisfying certain conditions in order to get an integrable equation (\ref{Lie2}).
Nevertheless, under the assumed restrictions, we may  obtain a
 system of differential equations admitting no solution. As an
 application,
we  show that many known results can be
recovered  and explained in this way.

We have already showed that the Riccati equations (\ref{ricceq}) with either $b_0\equiv 0$ or $b_2\equiv 0$ are reducible to linear
differential equations and therefore they are always
integrable. Hence, they are not interesting in our study and we focus our attention on reducing a Riccati equation (\ref{ricceq}), with $b_0b_2\ne 0$ in an open
interval in $t$, into an integrable one by means of the action of a curve in $SL(2,\mathbb{R})$. With this aim,
 we consider the family of curves in $SL(2,\mathbb{R})$ with  $\beta=0$ and $\gamma=0$, i.e. we take curves of the form
$$\bar A(t)=\left(\begin{matrix}\alpha(t)&0\\0&\delta(t)\end{matrix}\right)\in SL(2,\mathbb{R})\,,\quad\alpha(t)\delta(t)=1.$$
 We already pointed out that a curve $\bar A(t)$ in $SL(2,\mathbb{R})$ induces a $t$-dependent change of variables in $\bar{\mathbb{R}}$ given by $y'(t)=\Phi(\bar A(t),y(t))$. In view of (\ref{Action}) and as $\alpha\delta=1$, we get that, in our case, such a change of variables is given by
\begin{equation}\label{yprime}
y'=\alpha^2(t)y=G(t)y\,,\quad G(t)\equiv \frac{\alpha(t)}{
\delta(t)}>0.
\end{equation}
In view of the relations (\ref{trans}), the initial Riccati equation is transformed by means of the curve $\bar A(t)$ into the new Riccati equation with $t$-dependent coefficients
$$b'_2=\delta^2\,b_2\,,\qquad b'_1=\alpha\,\delta\,b_1+\dot \alpha\,\delta-\alpha\,\dot
\delta\,,\qquad b'_0=\alpha^2\, b_0.$$
Furthermore, the functions $\alpha$ and $\delta $ are solutions of
 system (\ref{FS}), which in this case reads
\begin{equation}\label{RLFS}
\left(\begin{matrix}
\dot\alpha\\
0\\
0\\
\dot\delta
\end{matrix}\right)
=\left(\begin{matrix}
\frac{b'_1-b_1}{2}&b_2 &b'_0&0\\
-b_0& \frac{b'_1+b_1}{2}&0 &b'_0\\
-b'_2&0 &-\frac{b'_1+b_1}{2}& b_2\\
0&-b_2' &-b_0& -\frac{b'_1-b_1}{2}
\end{matrix}\right)\left(
\begin{matrix}
\alpha\\
0\\
0\\
\delta
\end{matrix}\right).
\end{equation}
The existence of particular solutions for the above system related to elements of $SL(2,\mathbb{R})$ and satisfying the required conditions determines integrability conditions for Riccati equations by the described method. Thus, let us analyse the existence of such solutions to get these integrability conditions.

From some of the relations of the system (\ref{RLFS}), we get that
$$-b_0\,\alpha+b'_0\, \delta=0\,,\qquad -b_2'\,\alpha+b_2\,\delta=0.$$
As $\alpha(t)\delta(t)= 1$, the above relations imply that $b_0b_2= b'_0b'_2$ and 
\begin{equation*}
\alpha^2=\frac{b_0'}{b_0}=\frac{b_2}{b_2'}\equiv G>0\,.
\end{equation*}
Hence, the transformation formulas (\ref{trans}) reduce to
\begin{equation}
b'_2=\alpha^{-2}\,b_2\,,\qquad b'_1=b_1+2\frac{\dot \alpha}\alpha
\,,\qquad b'_0=\alpha^2 b_0\,.\label{transfb}
\end{equation}

Then, in order to exist a $t$-dependent function $D$ and two real constants
$c_0$ and $c_2$, with $c_0c_2\neq 0$, such that
$b'_2=Dc_2$ and $b'_0=Dc_0$, the function $D$ must be given by
\begin{equation*}
D^2c_0c_2=b_0b_2\Longrightarrow D=\pm\sqrt{\frac{b_0b_2}{c_0c_2}}\,,
\end{equation*}
where we have used that $b'_0b'_2=b_0b_2$. On the other hand, as $b'_0/b_0=\alpha^2>0$, we have to fix the sign $\kappa$ of the function $D$ in order to satisfy this relation, i.e. ${\rm sg}(c_0D)={\rm sg}(b_0)$. Therefore,
$$
\kappa={\rm sg}(D)={\rm sg}(b_0/c_0).
$$
Also, as $b_0b_2=b'_0b'_2$, we get that ${\rm sg}(b_0b_2)={\rm sg}(c_0c_2D^2)={\rm sg}(c_0c_2)$. Furthermore, in view of the relations (\ref{transfb}), $\alpha$ is determined, up to a sign, by
\begin{equation}\label{otroalfa}
\alpha=\sqrt{\frac{Dc_0}{b_0}}=\left(\frac{c_0}{c_2}\,\frac{b_2}{b_0} \right)^{1/4}\,,
\end{equation}
and therefore the change of variables (\ref{yprime}) reads:
\begin{equation}\label{Chang}
y'=\frac{D(t)c_0}{b_0(t)}y\,.
\end{equation}

Finally, as a consequence of (\ref{transfb}), in order for $b'_1$ to be the product $b'_1=c_1\, D$, we see that
\begin{equation}\label{eq10}
b_1+2\, \frac{\dot \alpha}{\alpha}=\kappa c_1 \sqrt{\frac{b_0b_2}{c_0c_2}}\,.
\end{equation}
Using (\ref{otroalfa}) we get
$$4\,\frac{\dot
  \alpha}{\alpha}=\frac 1{\alpha^4}\, \frac {d\alpha^4}{dt}=\frac{b_0}{b_2}\,\frac
d{dt}\left(\frac{b_2}{b_0}\right)=\frac{b_0}{b_2}\,\, \frac{\dot b_2b_0-\dot
  b_0 b_2}{b_0^2}=\frac{\dot b_2}{b_2}-\frac{\dot b_0}{b_0},
$$
and replacing $2\dot\alpha/\alpha$ in (\ref{eq10}) for the value obtained above, we see that the
required integral condition is
\begin{equation*}
\sqrt{\frac{c_0c_2}{b_0b_2}}\left[b_1+\frac{1}{2}\left(\frac{\dot b_2}{b_2}-\frac{\dot b_0}{b_0}\right)\right]=\kappa c_1\,.
\end{equation*}

Conversely, it can be verified that if the above integrability condition holds and $D^2c_0c_2=b_0b_2$, then the change of variables (\ref{Chang}) transforms the Riccati equation (\ref{ricceq}) into $dy'/dt=D(t)(c_0+c_1y'+c_2y'^2)$, with $c_0c_2\neq 0$.

In summary:

\begin{theorem}\label{TU} The necessary and sufficient condition
for the existence of a  transformation
\begin{equation*}
y'=G(t)y,\quad G(t)>0,
\end{equation*}
 relating the Riccati equation
\begin{equation*}
\frac{dy}{dt}=b_0(t)+b_1(t)y+b_2(t)y^2\,,  \qquad b_0b_2\ne 0,
\end{equation*}
to an integrable one given by
\begin{equation}
\frac{dy'}{dt}=D(t)(c_0+c_1y'+c_2y'^2)\,,\quad c_0c_2\neq 0\label{eqDcs}
\end{equation}
where $c_0, c_1, c_2$ are real numbers and $D(t)$ is a non-vanishing function, are
\begin{equation}
D^2c_0c_2=b_0b_2,\qquad \left(b_1+\frac{1}{2}\left(\frac{\dot b_2}{b_2}-\frac{\dot b_0}{b_0}\right)\right)\sqrt{\frac{c_0c_2}{b_0b_2}}=\kappa c_1,\label{DinTh2}
\end{equation}
where $\kappa={\rm sg}(D)=sg(b_0/c_0)$. The  transformation is then uniquely defined by
\begin{equation*}
y'=\sqrt{\frac{b_2(t)c_0}{b_0(t)c_2}}\,y\,.
\end{equation*}
\end{theorem}

As a consequence of Theorem \ref{TU}, given a Riccati equation
\begin{equation*}
\frac{dy}{dt}=b_0(t)+b_1(t)y+b_2(t)y^2\,, \qquad b_0(t)b_2(t)\ne 0,
\end{equation*}
if there are real constants $c_0,c_1$ and $c_2$, with $c_0c_2\neq 0$, such that
\begin{equation*}
\sqrt{\frac{c_0c_2}{b_0b_2}}\left(b_1+\frac{1}{2}\left(\frac{\dot b_2}{b_2}-\frac{\dot b_0}{b_0}\right)\right)=\kappa c_1,
\end{equation*}
there exists a $t$-dependent linear change of variables transforming the given equation into
 an integrable Riccati equation of the form
\begin{equation}
\frac{dy'}{dt}=D(t)(c_0+c_1y'+c_2y'^2), \qquad c_0c_2\neq 0,\label{ricc1dim}
\end{equation}
and the function $D$ is given by (\ref{DinTh2}) with the sign determined by $\kappa$. 

From the previous results, it can be derived the following corollary.
\begin{corollary}\label{CTU}
A Riccati equation (\ref{ricceq}) with  $b_0b_2\ne 0$ can be transformed into a Riccati equation of the form (\ref{ricc1dim}) by a $t$-dependent  change of variables $y'=G(t)y$, with $G(t)>0$, if and only if
\begin{equation}
\frac{1}{\sqrt{|b_0b_2|}}\left(b_1+\frac{1}{2}\left(\frac{\dot b_2}{b_2}-\frac{\dot
      b_0}{b_0}\right)\right)=K,
\label{resCor2}
\end{equation}
for a certain real constant $K$. In such a case, the Riccati equation (\ref{ricceq}) is integrable by quadratures.
\end{corollary}

According to Theorem \ref{TU}, if we start with the integrable Riccati
Eq. (\ref{ricc1dim}), we can
obtain the set of all Riccati equations that can be reached from it by means of a
transformation of the form (\ref{yprime}).
\begin{corollary}\label{C2TU} Given an integrable Riccati equation
\begin{equation*}
\frac{dy}{dt}=D(t)(c_0+c_1y+c_2y^2),\qquad c_0c_2\neq 0,
\end{equation*}
with $D(t)$ a non-vanishing function, the set of Riccati equations which can be  obtained with a
transformation $y'=G(t)y$, with $G(t)>0$, are
those of the form:
\begin{equation*}
\frac{dy'}{dt}=b_0(t)+\left( \frac{\dot b_0(t)}{b_0(t)}-\frac{\dot D(t)}{D(t)}+c_1D(t)\right) y'+\frac{D^2(t)c_0c_2}{b_0(t)}y'^2\,,
\end{equation*}
with 
$$G=\frac{Dc_0}{\sqrt{b_0}}\,.$$
\end{corollary}

Therefore starting with  an integrable equation  we can generate a family of
solvable Riccati equations whose coefficients are parametrised by a non-vanishing function $b_0$. Moreover, the integrability condition
to check whether a Riccati equation belongs to this family can be easily verified.

 These results can now be used for a better understanding  of
 some integrability conditions found  in the literature.

 \medskip

 $\bullet$  {\it The case of Allen and Stein}:

\medskip

 The results of the paper by Allan and Stein \cite{AS64} can be recovered through our general approach. In that work, a  Riccati equation (\ref{ricceq}) with $b_0b_2>0$ and $b_0$, $b_2$ differentiable functions satisfying the condition
\begin{equation}\label{ALintegra}
\frac{b_1+\frac{1}{2}\left(\frac{\dot b_2}{b_2}-\frac{\dot b_0}{b_0}\right)}{\sqrt{b_0b_2}}=C, 
\end{equation}
where $C$ is a real constant, was transformed into the integrable one
\begin{equation}\label{FREAS64}
\frac{dy'}{dt}=\sqrt{b_0(t)b_2(t)}\left(1+Cy'+y'^2\right),
\end{equation}
through the $t$-dependent  linear transformation
\begin{equation*}
y'=\sqrt{\frac{b_2(t)}{b_0(t)}}y\,.
\end{equation*}

If integrability condition (\ref{ALintegra}) is satisfied by a Riccati equation, such an equation also holds the assumptions of the Corollary \ref{CTU} and, therefore, the integrability condition given in Theorem \ref{TU} with
\begin{equation*}
c_0=1=c_2,\quad c_1=C,\quad D=\sqrt{b_0b_2}.
\end{equation*}
Consequently, the corresponding transformation given by Theorem \ref{TU} reads
\begin{equation*}
y'=\sqrt{\frac{b_2(t)}{b_0(t)}}y\,,
\end{equation*}
showing that the transformation in \cite{AS64} is a particular case of our results. This is not an unexpected result because Theorem \ref{TU} shows that if such a time-dependent change of variables is used to transform a Riccati equation (\ref{ricceq}) into one of the form (\ref{eqDcs}), this change of variables must be of the form (\ref{Chang}) and the initial Riccati equation must hold the integrability conditions (\ref{DinTh2}).

\medskip

$\bullet$ {\it The case of Rao and Ukidave}:

\medskip

Rao and Ukidave stated in their work \cite{RU68} that the Riccati equation (\ref{ricceq}), with $b_0b_2>0$, can
be transformed into
an integrable Riccati equation of the form
\begin{equation*}
\frac{dy'}{dt}=\sqrt{cb_0b_2}\left(1-ky'+\frac{1}{c}{y'}^2\right),
\end{equation*}
through a $t$-dependent  linear transformation
\begin{equation*}
y'=\frac{1}{v(t)}y,
\end{equation*}
if there exist  real constants $c$ and $k$ such that following integrability condition holds
\begin{equation}\label{CondRU1}
 b_2=\frac{b_0}{cv^2},
\end{equation}
with $v$ being a solution of the differential
equation
\begin{equation}\label{CondRU2}
\frac{dv}{dt}=b_1(t)v+kb_0(t)\,.
\end{equation}

Note that, in view of (\ref{CondRU1}), necessarily $c>0$ and if the integrability conditions (\ref{CondRU1}) and (\ref{CondRU2}) hold with constants $c$ and $k$ and a negative solution $v(t)$, the same conditions hold for the constants $c$ and $-k$ and a positive solution $-v(t)$. Consequently, we can restrict ourselves to studying the integrability conditions (\ref{CondRU1}) and (\ref{CondRU2})  for positive solutions $v(t)>0$. In such a case, the previous method uses a $t$-dependent  linear change of coordinates of the form (\ref{yprime}) and the final Riccati equation are of the type described in our work (\ref{eqDcs}), therefore
the integrability conditions derived by Rao and Ukidave must be a particular instance of the integrable cases provided by
 Theorem \ref{TU}. 

Using the value of $v(t)$ in terms of the constant $c$ and the functions $b_0$ and $b_2$ obtained from formula (\ref{CondRU1}) and Eq. (\ref{CondRU2}), we get that
\begin{equation*}
\frac{1}{\sqrt{|b_0b_2|}}\left(b_1+\frac 12\left(\frac{\dot b_2}{b_2}-\frac{\dot b_0}{b_0}\right)\right)=-k\,{\rm sg}(b_0)\sqrt{c}.
\end{equation*}
Hence, the Riccati equations obeying conditions (\ref{CondRU1}) and (\ref{CondRU2}) satisfy the integrability conditions of Corollary \ref{CTU}. Moreover, if we choose
 \begin{equation*}
D^2=cb_0b_2,\quad c_0=1,\quad c_1=-k,\quad c_2=c^{-1}\,,
\end{equation*}
 then $D=\sqrt{c b_0b_2}$ and the only possible transformation (\ref{yprime})  given by Theorem \ref{TU} reads
\begin{equation*}
y'=\alpha^2(t)y=\sqrt{\frac{cb_2(t)}{b_0(t)}}y,
\end{equation*}
and then
\begin{equation*}
\frac{1}{v}=\sqrt{\frac{cb_2}{b_0}}.
\end{equation*}
In this way, we recover one of the results derived by Rao and Ukidave in \cite{RU68}.
\medskip

$\bullet$ {\it The case of Kovalevskaya}:
\medskip

Kovalevskaya showed in the paper \cite{Ko06} that the Riccati equation
\begin{equation*}
\frac{dy}{dt}=F(t)+\left(L+\frac{\dot F(t)}{F(t)}\right)y-\frac{K}{F(t)}y^2,
\end{equation*}
where $K$ and $L$ are real constant, can be integrated through quadratures. It can be verified that the above family of Riccati equations
holds the assumption of Corollary \ref{CTU}. Indeed, taking $c_0=1$, $c_2=-K$, $c_1=L$ we get that $\kappa=1$,
\begin{equation*}
\sqrt{\frac{c_0c_2}{b_0b_2}}\left(b_1+\frac{1}{2}\left(\frac{\dot b_2}{b_2}-\frac{\dot
      b_0}{b_0}\right)\right)=L=c_1,
\end{equation*}
and $D=\sqrt{{b_2b_2}/{c_2c_0}}=1$. Therefore, Theorem \ref{TU} shows that the above family of Riccati equations can be integrated.
Moreover, taking the above values of the constants $c_0$, $c_1$, $c_2$ and the function $b_0(t)=F(t)$, Corollary \ref{C2TU} reproduces the family of Riccati equations analysed by Kovalevskaya.
\medskip

$\bullet$ {\it The case of Hong-Xiang}:

\medskip

As a final example we can consider the Riccati equation 
\begin{equation*}
\frac{dy}{dt}=- y^2-\left(2 b G(t)-\frac{\dot G(t)}{G(t)}\right) y-c G^2(t),
\end{equation*}
used to analyse a certain integrable linear differential equation in \cite{Ro07} which was also analysed by Hong-Xiang \cite{HX82}. The above Riccati equation
satisfies the integrability condition (\ref{resCor2}) and hence it can be integrated. Indeed, we have
$$
\left\{
\begin{aligned}
b_0(t)&=-c G^2(t),\\
b_1(t)&=-\left(2bG(t)-\frac{\dot G(t)}{G(t)}\right),\\
b_2(t)&=-1,
\end{aligned}\right.
$$
and therefore we get that
$$
\frac{b_1+\frac{1}{2}\left(\frac{\dot b_2}{b_2}-\frac{\dot b_0}{b_0}\right)}{\sqrt{|b_0b_2|}}={\rm const.}
$$
In summary, many integrability conditions shown in the literature are equivalent to or
particular instances of those given in our more general statements.

\section{Integrability and reduction}
\indent

In this Section we develop a  procedure that is similar
 to the one derived throughout the previous Sections but we here consider solutions of system
 (\ref{FS}) in two-parameter subsets of $SL(2,\mathbb{R})$. In this case, we recover
 some known integrability conditions, e.g. a certain kind of integrability used in \cite{CarRamGra}. More specifically, we try to relate a Riccati equation
(\ref{ricceq}) to an integrable one associated, as a Lie system, with a curve of the form
 ${\rm a}'(t)=-D(t)(c_0{\rm a}_0+c_1{\rm a}_1+c_2{\rm a}_2)$, with $c_2\neq 0$ and a non-vanishing function $D=D(t)$. Furthermore, we consider solutions of system  (\ref{Sys}) with $\gamma=0$ and $\alpha>0$ related to elements of $SL(2,\mathbb{R})$, i.e. we analyse transformations 
$$y'=\frac{\alpha(t)}{\delta(t)}y+\frac{\beta(t)}{\delta(t)}=\alpha^2(t)\,y+
\alpha(t)\beta(t)\,.$$
In this case, using the expression in coordinates (\ref{FS}) of system (\ref{Sys}), we get that
\begin{equation}\label{PC}
\left(\begin{matrix}
\dot\alpha\\
\dot\beta\\
0\\
\dot\delta
\end{matrix}\right)=\left(
\begin{matrix}
\frac{b'_1-b_1}{2}&b_2 &b'_0&0\\
-b_0& \frac{b'_1+b_1}{2}&0 &b'_0\\
-b'_2&0 &-\frac{b'_1+b_1}{2}& b_2\\
0&-b_2' &-b_0& -\frac{b'_1-b_1}{2}
\end{matrix}\right)\left(
\begin{matrix}
\alpha\\
\beta\\
0\\
\delta
\end{matrix}\right)\,,
\end{equation}
where $b'_j=D\,c_j$ and $c_j\in\mathbb{R}$ for $j=0,1,2$. As we suppose $b'_2\neq 0$, the third equation of the above system implies
\begin{equation*}
\frac{\alpha}{\delta}=\frac{b_2}{b_2'}.
\end{equation*}
As $\alpha\delta=1$ in order to obtain a solution of (\ref{Sys}) related to an element of $SL(2,\mathbb{R})$ and $b_2'=D c_2$,
 we get
\begin{equation}\label{Drelation}
\alpha^2=\frac{b_2}{D c_2}.
\end{equation}
Hence, $\alpha$ is determined, up to a sign, by the values of $b_2(t), D$ and $c_2$. In this way, if we take $\alpha$ to be positive, the first differential equation of system (\ref{PC}) gives us
 the value of $\beta$ in terms of the related initial and final Riccati equation, i.e.
$$
\beta=\frac{1}{b_2}\left(\dot \alpha-\frac{b'_1-b_1}{2}\alpha\right).
$$
Taking into account the relation (\ref{Drelation}), the above expression is equivalent to the differential equation 
\begin{equation*}
\frac{dD}{dt}=\left(b_1(t)+\frac{\dot b_2(t)}{b_2(t)}\right)D-c_1D^2-2b_2(t)D\beta \left(\frac{c_2D}{b_2(t)}\right)^{1/2},
\end{equation*}
and, as $\alpha\delta= 1$, we can define $M=\beta/\alpha$ and rewrite the above expression as follows
\begin{equation*}
\frac{dD}{dt}=\left(b_1(t)+\frac{\dot b_2(t)}{b_2(t)}\right)D-c_1D^2-2b_2(t)MD.
\end{equation*}
Considering the differential equation in $\dot \beta$ in terms of $M$, we get the equation
\begin{equation*}
\frac{dM}{dt}=-b_0(t)+\frac{c_0c_2}{b_2(t)}D^2+b_1(t) M-b_2(t) M^2\,.
\end{equation*}
Finally, as $\delta\alpha=1$ is a first-integral of system (\ref{Sys}), if the system for the variables $M$ and $D$ and all the obtained conditions are satisfied, the value $\delta=\alpha^{-1}$ satisfies its corresponding differential equation of the system (\ref{PC}). To sum up, we have obtained  the following  result.
\begin{theorem}\label{FT2} Given a  Riccati equation (\ref{ricceq}),
there exists a transformation
\begin{equation*}
y'=G(t)y+H(t)\,,\qquad G(t)>0\,,
\end{equation*}
 relating it to the integrable equation
\begin{equation}\label{fequation}
\frac{dy'}{dt}=D(t)(c_0+c_1y'+c_2y'^2),
\end{equation}
with $c_2\neq 0$ and $D$ a non-vanishing function, if and only if there exist functions $D$ and $M$ satisfying the following system
\begin{eqnarray*}
\left\{\begin{aligned}
\frac{dD}{dt}&=\left(b_1(t)+\frac{\dot b_2(t)}{b_2(t)}\right)D-c_1D^2-2b_2(t)MD,\\
\frac{dM}{dt}&=-b_0(t)+\frac{c_0c_2}{b_2(t)}D^2+b_1(t) M-b_2(t) M^2.
\end{aligned}\right.
\end{eqnarray*}
The transformation is then given by
\begin{equation}\label{ChangeT3}
y'=\frac{b_2(t)}{D(t)c_2}(y+M(t))\,.
\end{equation}
\end{theorem}

Consider $c_0=0$ in Eq. (\ref{fequation}). Thus, the system determining the curve
in $SL(2,\mathbb{R})$ performing the transformation of Theorem \ref{FT2} is
\begin{equation}
\left\{
\begin{array}{rcl}
\dfrac{dD}{dt}&=&\left(b_1(t)+\dfrac{\dot b_2(t)}{b_2(t)}\right)D-c_1D^2(t)-2b_2(t)MD,\\
\dfrac{dM}{dt}&=&-b_0(t)+b_1(t) M-b_2(t) M^2.
\end{array}\label{RedSep}\right.
\end{equation}
On one hand, this system does not involve any integrability condition because, as a consequence of the Theorem of existence and uniqueness of solutions,
there always exists a solution for every initial condition. On the other hand, such solutions can  be as difficult to be found as the general solution of the initial
Riccati equation. Hence, in order to find a particular solution, we need to look for some
simplifications. For instance, we can consider the case in which $M=b_1/b_2$. In this case, the first differential equation of the above system does not depend on $M$ and reads
$$
\frac{dD}{dt}=\left(-b_1(t)+\frac{\dot b_2(t)}{b_2(t)}\right)D-c_1D^2
$$
and it is integrable by quadratures. Its solution reads
\begin{equation*}
D(t)=\frac{\exp\left(\int_0^t A(t')dt'\right)}{C+c_1\int^t_0\exp\left(\int_0^{t''} A(t')dt'\right)dt''}\,,\qquad A(t)=\left(-b_1(t)+\frac{\dot b_2(t)}{b_2(t)}\right).
\end{equation*}

Meanwhile, the condition for  $M=b_1/b_2$ to be a solution of the second equation in
(\ref{RedSep}) is
\begin{equation*}
\frac{d}{dt}\left(\frac{b_1}{b_2}\right)=-b_0\,,
\end{equation*}
giving rise to an integrability condition. This summarises one of the integrability conditions considered in \cite{Ra62}. 

Next, we recover from this new viewpoint the well-known result that the knowledge of a particular solution of the Riccati equation allows us to solve the system (\ref{RedSep}). In fact, under the change of variables $M\mapsto -y$, system (\ref{RedSep}) becomes
\begin{eqnarray}\label{eq8}
\left\{\begin{aligned}
\frac{dD}{dt}&=\left(b_1(t)+\frac{\dot b_2(t)}{b_2(t)}\right)D-c_1D^2+2b_2(t)yD,\\
\dfrac{dy}{dt}&=b_0(t)+b_1(t) y+b_2(t) y^2.
\end{aligned}\right.
\end{eqnarray}
Note that each particular solution of the above system is the form $(D_p(t),y_p(t))$, with $y_p(t)$ a particular solution of the Riccati equation (\ref{ricceq}). Therefore, given such a particular solution $y_p(t)$, the function $D_p=D_p(t)$, corresponding to the particular solution $(D_p(t),y_p(t))$ of system (\ref{eq8}), holds the integrable equation 
\begin{equation}\label{PS}
\frac{dD_p}{dt}=\left(b_1(t)+\frac{\dot b_2(t)}{b_2(t)}+2b_2(t)y_p(t)\right)D_p-c_1D_p^2.
\end{equation}
Hence, the knowledge of a particular solution $y_p(t)$ of the Riccati equation (\ref{ricceq}) enables us to get a particular solution $(D_p(t),y_p(t))$ of system (\ref{eq8}) and, taking into account the change of variables $y\mapsto -M$, a particular solution $(D_p(t),M_p(t))=(D_p(t),-y_p(t))$ of system (\ref{RedSep}). Finally, the  functions $M_p(t)$ and $D(t)$ determines a change of variables (\ref{ChangeT3}) given by Theorem \ref{FT2} transforming the initial Riccati equation (\ref{ricceq}) into another one related, as a Lie system, to a solvable Lie algebra of vector fields. In this way, we describe a reduction process similar to that one pointed out in \cite{CarRamGra}. Nevertheless, we here directly obtain a reduction to a Riccati equation related, as a Lie system, to a one dimensional Lie subalgebra of $\mathfrak{sl}(2,\mathbb{R})$ through one of its particular
solutions.

There exists many ways to impose conditions on the coefficients of the second
equation of (\ref{eq8}) for being able  to obtain one of its particular
solutions easily. Now, we give some particular examples of this.

If there exists a real constant $c$ such that for the time-dependent functions $b_0$, $b_1$ and $b_2$ we have that $b_0+b_1 c+b_2 c^2=0$, then $c$ is a particular solution. This resumes some cases found in
\cite{CarRamGra, Stre}. For instance:

\begin{enumerate}
 \item $b_0+b_1+b_2=0$ means that $c=1$ is a particular solution.
\item $c_1^2b_0+c_1c_2b_1+c_2^2b_2=0$ means that $c=c_2/c_1$ is a particular solution.
\end{enumerate}

In these particular instances, we can find $D$ through the first differential
equation of (\ref{eq8}).

As a first application of this last case  we can integrate the Riccati equation
\begin{equation}\label{Hovy}
\frac{dy}{dt}=-\frac{n}{t}+\left(1+\frac{n}{t}\right)y-y^2.
\end{equation}
related to Hovy's equation \cite{Ro07}. This Riccati equation admits the particular constant solution $y_p(t)=1$. Using such a particular solution in Eq. (\ref{PS}) and fixing, for instance, $c_1=0$, we can obtain a particular solution for Eq. (\ref{PS}), e.g. $D_p(t)=t^ne^{-t}$. Therefore, we get that $(t^ne^{-t},1)$ is a solution of the system (\ref{eq8}) related to Eq. (\ref{Hovy}) and $(t^ne^{-t},-1)$ is a solution of the system (\ref{RedSep}). In this way, Theorem \ref{FT2} states that the transformation (\ref{ChangeT3}), determined by the $D_p(t)=t^ne^{-t}$ and $M_p(t)=-1$, of the form
\begin{equation}\label{rel}
y'=-t^{-n}e^tc_2^{-1}(y-1),
\end{equation}
relates the solutions of Eq. (\ref{Hovy}) to the integrable one 
$$
\frac{dy'}{dt}=e^{-t}t^n(c_0+c_2y'^2).
$$
If we fix $c_0=1$ and $c_2=1$ the solution for the above equation is
\begin{equation*}
y'(t)=-\frac{1}{-K+\Gamma(1+n,t)},
\end{equation*}
where $K$ is an integration constant and $\Gamma(a,b)$ is the incomplete Euler's Gamma function
\begin{equation*}
\Gamma(a,t)=\int^\infty_t t'^{a-1}e^{-t'}dt'.
\end{equation*}
In view of the change of variables (\ref{rel}), the solutions $y(t)$ of the
Riccati equation (\ref{Hovy}) 
and $y'(t)$ are related by means of the expression
$y'(t)=-t^{-n}e^tc_2^{-1}(y(t)-1)$. Therefore, if we
 substitute the general solution $y'(t)$ in this expression we can derive the
 general solution 
for the Riccati equation (\ref{Hovy}), that is, 
\begin{equation*}
y(t)=1-\frac{e^{-t}t^n}{\Gamma(n+1,t)+K}.
\end{equation*}

Another approach that can be summarised by Theorem \ref{FT2} is the factorisation
method developed in \cite{Ro07} to explain an integrability process for second-order
differential equations. In that work, it was analysed the differential equation:
\begin{equation}\label{SOE}
\frac{d^2y}{dt^2}+2P(t)\frac{dy}{dt}+\left(\frac{dP}{dt}+P^2(t)-\frac{d\phi}
{dt}-\phi^2(t)\right)y=0\,.
\end{equation}

We know that invariance under dilations leads to consider an adapted
variable $z$, such that $y=e^z$. Under this change of variables the equation obtained for $\psi=\dot z$ is the
Riccati equation
\begin{equation}\label{FRE}
\frac{d\psi}{dt}=-\psi^2-2P(t)\psi-\left(\frac{dP}{dt}(t)+P^2(t)-\frac{d\phi}
{dt}(t)-\phi^2(t)\right).
\end{equation}
This equation was integrated through a factorisation method in
\cite{Ro07}. Nevertheless, we can also integrate
 this equation if we take into account that $\psi_p(t)=\phi(t)-P(t)$ is
a particular solution of the above
differential equation and then applying the same procedure as for Eq.
(\ref{Hovy}). Indeed, 
as $\psi_p(t)$ is a particular  solution for the Riccati equation (\ref{FRE}),
we can obtain a 
particular solution $D_p=D_p(t)$ for Eq. (\ref{PS})  and by means of
the functions 
$M_p(t)=-\psi_p(t)$ and $D_p(t)$ we can obtain the solution of the Riccati
equation (\ref{FRE}). 
Finally, inverting the change of variables used to relate Eq.
(\ref{SOE}) to (\ref{FRE}) 
we obtain the solution for Eq. (\ref{SOE}).

\section{Linearisation of Riccati equations}
\indent 

One can also study the problem of the linearisation of Riccati equations
through the linear fractional
 transformations (\ref{yprime}). This set
 of time-dependent transformations is general enough to include many of the
time-dependent or time-independent changes of variables already used to study Riccati equations, e.g. it
 allows us to recover the results of
\cite{RDM05}. As a main result, we state in this Section some integrability
conditions to be able to 
transform a $t$-dependent
Riccati equation into a linear one by means of a diffeomorphism on
$\overline{\mathbb{R}}$ associated with 
certain linear fractional transformations.

As a first insight in the linearisation process, note that Corollary \ref{CorCur} shows that
 there always exists a curve in $SL(2,\mathbb{R})$, and then a
 $t$-dependent
 linear fractional
 transformation on ${\overline{\mathbb{R}}}$, transforming a given Riccati
 equation into any other one. 
In particular, if we fix $b_2'=0$ in the final Riccati equation, we obtain that there is
 a $t$-dependent linear fractional change of variables  transforming any
 Riccati equation (\ref{ricceq}) 
into a linear one.
Nevertheless, as the Lie system (\ref{Sys}) describing such a transformation is
not related to a 
solvable Lie algebra of vector fields, it is not easy to find such a transformation in the general case.

Let us try to relate a Riccati equation (\ref{ricceq}) to a linear differential
equation by means of 
a linear fractional transformation (\ref{Action}) determined by a vector
$(\alpha,\beta,\gamma,\delta)\in \mathbb{R}^4$ 
with $\alpha\delta-\beta\gamma=1$. In this case, the existence of solutions of
the system (\ref{FS}) performing 
such a transformation is an easy task and we can look for integrability
conditions to get the corresponding 
change of variables. Note that as $(\alpha,\beta,\gamma,\delta)$ is a constant,
we have 
$\dot\alpha=\dot\beta=\dot\gamma=\dot\delta=0$ and, in view of (\ref{FS}), the
diffeomorphism on 
$\overline{\mathbb{R}}$ performing the transformation is related to a vector in the kernel of the matrix
\begin{equation}\label{EM}
B=\left(\begin{matrix}
\frac{b'_1-b_1}{2}&b_2 &b'_0&0\\
-b_0& \frac{b'_1+b_1}{2}&0 &b'_0\\
0&0 &-\frac{b'_1+b_1}{2}& b_2\\
0&0 &-b_0& -\frac{b'_1-b_1}{2}
\end{matrix}\right),
\end{equation}
where we assume $b_0b_2\neq 0$ in an open interval in the variable $t$. We
leave out the study of the case 
$b_0b_2=0$  in an open interval because, as it was shown in Sec.
\ref{IntRicEqu}, this case is 
known to be integrable.

The necessary and sufficient condition for a non-trivial $\ker B$ is $\det B=
0$ and, therefore, a 
short calculation shows that $\dim\,\, {\rm ker}\, B>0$ if and only if
$(-b_1^2+b_1'^2(t)+4 b_0 b_2)^2=0.$ Thus, $b_1'=\pm \sqrt{b_1^2-4 b_0 b_2}$ and
$b_1'$ is fixed, but a sign,
 by the values of $b_0$, $b_1$ and $b_2$. Let us study the kernel of the matrix
 $B$ in the positive 
and negative cases for $b'_1$.

$\bullet$ Positive case:
The kernel of the matrix (\ref{EM}) is given by the vectors

\begin{equation*}
\left(\delta\frac{b_0'}{b_0}+\beta\frac{b_1+\sqrt{b_1^2-4 b_0 b_2}}{2 b_0}
,\beta,-\delta\frac{-b_1+\sqrt{b_1^2-4 b_0b_2}}{2 b_0},\delta\right), \qquad \delta,\beta\in\mathbb{R}.
\end{equation*}
Recall that we only consider the constant elements of $\ker B$, therefore there
should be two real 
constants $K_1$ and $K_2$ such that
\begin{eqnarray*}
K_1=\delta\frac{b_0'}{b_0}+\beta\frac{b_1+\sqrt{b_1^2-4 b_0 b_2}}{2 b_0},\qquad K_2=\frac{-b_1+\sqrt{b_1^2-4 b_0b_2}}{2 b_0}.
\end{eqnarray*}
Moreover, in order to relate these vectors to elements in $SL(2,\mathbb{R})$ we
have to 
impose that $\det (K_1,\beta,-\delta K_2,\delta)=\delta(K_1 +\beta K_2)=1$. 

The second condition imposes a restriction on the coefficients of the
initial Riccati equation to be linearisable  by a constant linear fractional
transformation (\ref{Action}). 
Then, if this condition is satisfied we can fix $\beta,\gamma, K_1$ and $b_0'$
to satisfy the other conditions. 
Thus, the only linearisation  condition  is the condition on $K_2$.

$\bullet$ Negative case: In this case, $\ker\,B$ reads
\begin{equation*}
\left(\delta \frac{b_0'}{b_0}+\beta\frac{b_1-\sqrt{b_1^2-4 b_0 b_2}}{2 b_0}
,\beta,-\delta\frac{-b_1-\sqrt{b_1^2-4 b_0b_2}}{2 b_0},\delta\right),\qquad \delta,\beta\in\mathbb{R},
\end{equation*}
and now the new conditions reduce to the existence of two real constants $K_1$ and $K_2$ such that 
\begin{eqnarray*}
K_1=\delta \frac{b_0'}{b_0}+\beta\frac{b_1-\sqrt{b_1^2-4 b_0 b_2}}{2
  b_0},\qquad K_2=
\frac{-b_1-\sqrt{b_1^2-4 b_0b_2}}{2 b_0},
\end{eqnarray*}
with $\delta(K_1+\beta  K_2)=1$. If the condition in $K_2$ is satisfied we can
proceed as in 
the positive case to obtain the transformation performing the linearisation of
the 
initial Riccati equation.

In summary:
\begin{theorem}
The necessary and sufficient condition for the existence of a diffeomorphism on
$\bar{\mathbb{R}}$ of linear 
fractional type associated with a transformation
of $SL(2,\mathbb{R})$ transforming the Riccati equation (\ref{ricceq}) into a
linear differential equation 
is the existence of a real constant $K$ such that
\begin{equation}\label{IntCond}
K=\frac{-b_1\pm\sqrt{b_1^2-4 b_0b_2}}{2 b_0}.
\end{equation}
\end{theorem}

As a Riccati equation (\ref{ricceq}) holds condition (\ref{IntCond}) if and
only if $K$ is a constant 
particular solution, we get the following corollary:

\begin{corollary}
A Riccati equation can be linearised by means of a diffeomorphism on
$\overline{\mathbb{R}}$ of the form 
(\ref{Action}) if and only if it admits a constant particular solution.
\end{corollary}

Ibragimov showed that a Riccati equation (\ref{ricceq}) is linearisable by
means of a change 
of variables $z=z(y)$ if and only if the Riccati equation admits a constant
solution \cite{Ib08}. 
Additionally, we have proved that in such a case, the change of variables
can be described by 
means of a transformation of the type (\ref{Action}). 

Now, it can be checked that the example given in \cite{Ib08,RDM05} satisfies
the above integrability 
condition. In this work, the differential equation
\begin{equation*}
\frac{dy}{dt}=P(t)+Q(t)y+k(Q(t)-kP(t))y^2,
\end{equation*}
was studied. The only interesting case is that with $k\neq 0$ because the other
ones are linear. 
In this latter case, $b_0(t)=P(t)$, $b_1(t)=Q(t)$ and $b_2(t)=k(Q(t)-kP(t))$. Hence,
\begin{eqnarray*}
\begin{aligned}
\frac{-b_1-\sqrt{b_1^2-4 b_0b_2}}{2 b_0}=-k,
\end{aligned}
\end{eqnarray*}
and the integrability condition (\ref{IntCond}) holds. Now we may fix $K_1=0$
and we look for a 
solution for the condition $\det(K_1,\beta,-\delta K_2,\delta)=1$ reading
$k\delta\beta=-1$. 
As $k\neq 0$, we can take $\beta=-1/k$ to get from the above condition that
$\delta=1$. 
Thus the transformation is that one associated with the vector $(0,-1/k,k,1)$,
i.e. 
the linear fractional transformation
\begin{equation*}
y'=\frac{-1/k}{ky+1}
\end{equation*}
that is the same found in \cite{RDM05}.
In this way we only have to obtain $b_0'$ from the condition
$$
K_1=0=\delta\frac{b_0'}{b_0}+\beta\frac{b_1+\sqrt{b_1^2-4 b_0 b_2}}{2 b_0},
$$ to get the final linear differential equation, that is,
\begin{equation*}
\frac{dy'}{dt}=\frac{Q(t)}{P(t)}+(Q(t)-2P(t)k)y',
\end{equation*}
as it appears in \cite{CRL07}.
\section{Conclusions and outlook}
\indent 

It has been shown that previous works about the integrability of the Riccati
equation can be explained from the unifying viewpoint of Lie systems. The
transformations 
used in the study of the integrability condition for these equations have been
understood as 
induced by curves in $SL(2,\mathbb{R})$.

We have investigated a Lie system characterising the  $t$-dependent fractional
transformations 
relating different Riccati equations associated, as Lie systems,
with curves in $\mathfrak{sl}(2,\mathbb{R})$. We have used this differential equation
 and considered some simple instances. These simplifications have been
 used to analyse known integrability conditions and provide new ones.

We have also shown that the system (\ref{FS}) is a good way to describe
linear fractional time--dependent transformations and found necessary and sufficient
conditions for the linearisability or simplification of a Riccati equation
through time-independent 
and time-dependent transformations obtained from curves in $SL(2,\mathbb{R})$.

There  are many ways of simplifying (\ref{Sys}) and some of them  have been
developed here. Other ways can be used to obtain new
integrability conditions.

Finally, the theory used here can be extended to any other Lie system to
provide new or 
recover known integrability conditions. This fact is to be developed in forthcoming works.
\section*{Acknowledgements}
\indent

 Partial financial support by research projects MTM2009-11154, MTM2009-08166-E
 and E24/1 (DGA)
 are acknowledged. JdL also acknowledges
 a F.P.U. grant from  Ministerio de Educaci\'on y Ciencia.

\end{document}